\documentclass[intlimits,tbtags]{article}
\usepackage[utf8]{inputenc}
\usepackage{mathtools}

\makeatletter

\usepackage{amsthm}
\usepackage{mathtools}
\usepackage{amsfonts}
\newcommand*{\id}{{\normalfont\hbox{1\kern-0.4em1}}}
\newcommand{\hide}[1]{}

\newcommand*{\restricted}[1]{{\mid_{#1}}}

\newcommand*{\curl}{\operatorname{curl}}
\newcommand*{\grad}{\operatorname{grad}}

\newcommand*{\dive}{\operatorname{div}}
\newcommand*{\Grad}{\operatorname{Grad}}

\newcommand*{\Div}{\operatorname{Div}}

\newcommand*{\ii}{\mathrm{i}}

\DeclareMathAccent{\Circ}{\mathalpha}{operators}{"17}
\newcommand{\interior}[1]{\Circ{#1}}

\renewcommand{\Re}{\operatorname{\mathfrak{Re}}}

\newcommand{\oi}[2]{\left]#1,#2 \right[}

\newcommand{\lci}[2]{\left[#1,#2 \right[}
\newcommand{\rci}[2]{\left]#1,#2 \right]}

\renewcommand{\tilde}{\widetilde}
\renewcommand*{\epsilon}{\varepsilon}
\renewcommand*{\rho}{\varrho}
\theoremstyle{plain}
\newtheorem{thm}{Theorem}[section]
  \theoremstyle{definition}
  
  \theoremstyle{remark}
  \newtheorem{rem}[thm]{Remark}
 \theoremstyle{definition}
  
  \theoremstyle{plain}
  
  \theoremstyle{plain}
  \newtheorem{lem}[thm]{Lemma}
  \theoremstyle{plain}
  \newtheorem*{lem*}{Lemma}
  \theoremstyle{plain}
  \newtheorem*{thm*}{Theorem}
  \theoremstyle{plain}
  
  \theoremstyle{remark}

\makeatother

\begin{document}
  \author{Rainer Picard}   
\title{On an Electro-Magneto-Elasto-Dynamic Transmission Problem.}   
\maketitle \abstract{We consider a coupled system of Maxwell's equations and the equations of elasticity, where the coupling occurs not via material properties but through an  interaction on an interface separating the two regimes. Evolutionary well-posedness in the sense of Hadamard well-posedness supplemented by causal dependence is shown for a natural choice of generalized interface conditions. The results are obtained in a Hilbert space setting incurring no regularity constraints on the boundary and the interface of the underlying regions.}   

\section{Introduction}

Similarities between various initial boundary value problems of mathematical
physics have been noted as general observations throughout the literature.
Indeed, the work by K. O. Friedrichs, \cite{Ref166,CPA:CPA3160110306},
already showed that the classical linear phenomena of mathematical
physics belong – in the static case – to his class of \emph{symmetric
positive hyperbolic partial differential equations}, later referred
to as \emph{Friedrichs systems}, which are of the abstract form
\begin{equation}
\left(M_{1}+A\right)u=f,\label{eq:pssym}
\end{equation}
with $A$ at least formally, i.e. on $C_{\infty}$-vector fields with
compact support in the underlying region $\Omega$, a skew-symmetric
differential operator and the $L^{\infty}$-matrix-valued multiplication-operator
$M_{1}$ satisfying the condition 
\[
\mathrm{sym}\left(M_{1}\right)\coloneqq\frac{1}{2}\left(M_{1}+M_{1}^{*}\right)\geq c>0
\]
for some real number $c$. Indeed, a typical choice for the domain
of $A$ is to incorporate a boundary condition into $D\left(A\right)$,
so that $A$ is skew-selfadjoint ($A$ quasi-m-accretive would be
sufficient). Problem \eqref{eq:pssym} can be considered as the static
problem associated with the dynamic problem ($\partial_{0}$ denotes
the time-derivative) 
\begin{equation}
\partial_{0}M_{0}+M_{1}+A\label{eq:FS}
\end{equation}
with $M_{0}$ a selfadjoint $L^{\infty}$-multiplication-operator
and $M_{0}\geq0$, which were also addressed in \cite{CPA:CPA3160110306}.
It is noteworthy, that even the temporal exponential weight factor,
which plays a central role in our approach, is introduced as an ad-hoc
formal trick to produce a suitable $M_{1}$ for a well-posed static
problem. For the so-called time-harmonic case, where $\partial_{0}$
is replaced by $\ii\omega$, $\omega\in\mathbb{R}$, we replace $A$
simply by $\ii\omega M_{0}+A$ to arrive at a system of the form \eqref{eq:pssym}.

Operators of the Friedrichs type \eqref{eq:FS}, can be generalized
to obtain a fully time-dependent theory allowing for operator-valued
coefficients, indeed, in the time-shift invariant case, for systems
of the general form 
\begin{equation}
\tag{{Evo-Sys}}\left(\partial_{0}M\left(\partial_{0}^{-1}\right)+A\right)U=F\label{eq:problem-2}
\end{equation}
where $A$ is – for simplicity – skew-selfadjoint and $M$ an operator-valued
– say – rational function as an abstract coefficient. The meaning
of $M\left(\partial_{0}^{-1}\right)$ is in terms of a suitable function
calculus associated with the (normal) operator $\partial_{0}$, \cite[Chapter 6]{PDE_DeGruyter}.
We shall refer to such systems as evolutionary equations, \emph{evo-systems}
for short, to distinguish them from the special subclass of classical
(explicit) evolution equations.

In this paper we intend to study a particular transmission problem
between two physical regimes, elec\-tro-mag\-ne\-to-dy\-na\-mics
and elasto-dynamics, within this general framework and establish its
well-posedness, which for evo-systems entails not only Hadamard well-posedness,
i.e. \emph{uniqueness}, \emph{existence} and \emph{continuous dependence},
but also the crucial property of \emph{causality}. 

The peculiarity of the problem we shall investigate is that the interaction
between the two regimes is solely via the interface, not via material
interactions as in piezo-electrics, compare e.g. \cite{MMA:MMA3866}
for the latter type of effects.

After properly introducing evo-systems in the next section, we shall
establish the equations of electro-magneto-dynamics and elasto-dynamics
respectively, as such systems in Section \ref{sec:Maxwell's-Equations-and-Elastodynamics}.
Finally, in Section \ref{sec:An-Interface-Coupling} we establish
a particular interface coupling problem between the two regimes in
adjacent regions via a mother-descendant mechanism, see the survey
\cite{Picard2015}. We emphasize that our setup allows for arbitrary
open sets as underlying domains with no additional constraints on
boundary regularity.

\section{A Short Introduction to a Class of Evo-Systems}

\subsection{Basic Ideas}

We shall approach solving \eqref{eq:problem-2} by looking at the
equation as a space-time operator equation in a suitable Hilbert space
setting. Without loss of generality we may\footnote{Every complex Hilbert space $X$ is a real Hilbert space choosing
only real numbers as multipliers and 
\[
\left(\phi,\psi\right)\mapsto\Re\left\langle \phi|\psi\right\rangle _{X}
\]
as new inner product. Note that with this choice $\phi$ and $\ii\phi$
are always orthogonal. Moreover, for any skew-symmetric operator $A$
we have
\[
x\perp Ax
\]
for all $x\in D\left(A\right)$. 

Indeed, since $\left\langle x|y\right\rangle -\left\langle y|x\right\rangle =0$
(symmetry) we have
\[
\left\langle x|Ax\right\rangle -\left\langle Ax|x\right\rangle =0
\]
or by skew-symmetry
\begin{eqnarray*}
0 & = & \left\langle x|Ax\right\rangle -\left\langle Ax|x\right\rangle \\
 & = & 2\left\langle x|Ax\right\rangle 
\end{eqnarray*}
for all $x\in D\left(A\right)$. } and will assume that all Hilbert spaces are real.

Solutions will be discussed in a weighted $L^{2}$-space $H_{\nu}\left(\mathbb{R},H\right)$,
constructed by completion of the space $\interior C_{1}\left(\mathbb{R},H\right)$
of differentiable $H$-valued functions with compact support w.r.t.
$\left\langle \:\cdot\:|\:\cdot\:\right\rangle _{\nu,H}$ (norm: $\left|\:\cdot\:\right|_{\nu,H}$)
\[
\left(\varphi,\psi\right)\mapsto\int_{\mathbb{R}}\left\langle \varphi\left(t\right)|\:\psi\left(t\right)\right\rangle _{H}\:\exp\left(-2\nu t\right)dt.
\]

Here $H$ denotes a generic real Hilbert space. We introduce time-dif\-fe\-ren\-tia\-tion
$\partial_{0}$ as a closed operator in $H_{\nu}\left(\mathbb{R},H\right)$
defined as the closure of 
\begin{align*}
\interior C_{1}\left(\mathbb{R},H\right)\subseteq H_{\nu}\left(\mathbb{R},H\right) & \to H_{\nu}\left(\mathbb{R},H\right),\\
\varphi & \mapsto\varphi^{\prime}.
\end{align*}
The operator $\partial_{0}$ is normal in $H_{\nu}\left(\mathbb{R},H\right)$.
For $\nu_{0}\in\oi0\infty$, $\nu\in\oi{\nu_{0}}\infty$, we have
\begin{equation}
\mathrm{sym}\left(\partial_{0}\right)\coloneqq\overline{\frac{1}{2}\left(\partial_{0}+\partial_{0}^{*}\right)}=\nu\geq\nu_{0}>0,\label{eq:symdef}
\end{equation}
i.e. 
\[
\partial_{0}\mbox{ is a strictly (and uniformly w.r.t. \ensuremath{\nu\in\oi{\nu_{0}}\infty}) positive definite (i.e. m-accretive) operator. }
\]

This core observation can be lifted to a larger class of more complex
problems involving operator-valued coefficients and systems of the
general form 
\begin{equation}
\tag{{Evo-Sys}}\left(\partial_{0}M\left(\partial_{0}^{-1}\right)+A\right)U=F\label{eq:problem}
\end{equation}
where $A$ is – for simplicity – skew-selfadjoint and $M$ an operator-valued
– say – rational function as abstract coefficient. 

In many practical cases skew-selfadjointness of $A$ is evident from
its structure as a block operator matrix of the form
\[
A=\left(\begin{array}{cc}
0 & -C^{*}\\
C & 0
\end{array}\right),
\]
with $H=H_{0}\oplus H_{1}$ and $C:D\left(C\right)\subseteq H_{0}\to H_{1}$
a densely defined, closed linear operator.

\subsection{Well-Posedness for Evo-Systems.}

Since reasonable well-posedness requires closed operators we describe
our problem class more rigorously as of the form 
\begin{equation}
\tag{{Evo-Sys}}\overline{\left(\partial_{0}M\left(\partial_{0}^{-1}\right)+A\right)}U=F.\label{eq:problem-1}
\end{equation}
For a convenient special class, more than sufficient for our purposes
here, we record the following general well-posedness result, see \cite{Pi2009-1,PIC_2010:1889,Picard2015}.

\begin{thm}\label{thm:well}Let $z\mapsto M\left(z\right)$ be a
rational $\mathcal{L}\left(H,H\right)$-valued function in a neighborhood
of $0$ such that $M\left(0\right)$ is selfadjoint and\footnote{Here we use $\mathrm{sym}$ in an analogous meaning to \eqref{eq:symdef},
i.e.
\[
\mathrm{sym}\left(B\right)\coloneqq\frac{1}{2}\overline{\left(B+B^{*}\right)},
\]
which is equal to $\frac{1}{2}\left(B+B^{*}\right)$ since $B$ is
continuous.\label{fn:sym}}
\begin{equation}
\nu M\left(0\right)+\mathrm{sym}\left(M^{\prime}\left(0\right)\right)\geq\eta_{0}>0\label{eq:posdef}
\end{equation}
for some $\eta_{0}\in\mathbb{R}$ and all $\nu\in\oi{\nu_{0}}\infty$~,
$\nu_{0}\in\oi0\infty$ sufficiently large, and let $A$ be skew-selfadjoint.
Then well-posedness of \eqref{eq:problem} follows for all $\nu\in\oi{\nu_{0}}\infty$.
Moreover, the solution operator $\overline{\left(\partial_{0}M\left(\partial_{0}^{-1}\right)+A\right)}^{-1}$
is causal in the sense that 
\[
\chi_{_{\rci{-\infty}0}}\overline{\left(\partial_{0}M\left(\partial_{0}^{-1}\right)+A\right)}^{-1}=\chi_{_{\rci{-\infty}0}}\overline{\left(\partial_{0}M\left(\partial_{0}^{-1}\right)+A\right)}^{-1}\chi_{_{\rci{-\infty}0}}.
\]
\end{thm}

Indeed, apart from occasional side remarks we will simply have
\[
M\left(\partial_{0}^{-1}\right)=M_{0}+\partial_{0}^{-1}M_{1}
\]
and since $\partial_{0}$, $A$ can be \emph{continuously} extended
to suitable extrapolation spaces, it is justified\footnote{Albeit this being sometimes confusing and misleading, it is a common
practice in the field of partial differential equations. E.g. one
frequently writes 
\[
\Delta=\partial_{1}^{2}+\partial_{2}^{2}
\]
although $\phi\in D\left(\Delta\right)$ does in general not – as
the notation appears to suggest – allow for $\phi\in D\left(\partial_{1}^{2}\right)\cap D\left(\partial_{2}^{2}\right)$.} to drop the closure bar, which we shall do henceforth.

\section{\label{sec:Maxwell's-Equations-and-Elastodynamics}Maxwell's Equations
and the Equations of Linear Elasticity as Evo-Systems }

\subsection{Maxwell's Equations as an Evo-Systems.}

James Clerk Maxwell developed his new ideas on electro-magnetic waves
in 1861-64 resulting in his famous two volume publication: A Treatise
on Electricity and Magnetism, \cite{maxwell1873treatise}. His ingenious
contribution to what we nowadays call Maxwell's equations is to amend
Ampere's law with a so-called \emph{displacement current} term. Heaviside
and Gibbs have given the system in its now familiar form as 
\begin{align*}
\partial_{0}D+\sigma E-\curl H & =-j_{ext}\,,\;\mbox{(Ampere's law)}\\
\partial_{0}B+\curl E & =0,\;\mbox{(Faraday's law of induction)}\\
D & =\epsilon E,\\
B & =\mu H.
\end{align*}
The usually included divergence conditions are redundant, since the
two equations together with the material relations can be seen to
be leading already to a well-posed initial boundary value problem.
The so-called \emph{six-vector} block matrix \emph{form}: {\small{}
\[
\left(\partial_{0}\left(\begin{array}{cc}
\epsilon & 0\\
0 & \mu
\end{array}\right)+\left(\begin{array}{cc}
\sigma & 0\\
0 & 0
\end{array}\right)+\left(\begin{array}{cc}
0 & -\curl\\
\interior\curl & 0
\end{array}\right)\right)\left(\begin{array}{c}
E\\
H
\end{array}\right)=\left(\begin{array}{c}
-j_{ext}\\
0
\end{array}\right)
\]
}brings us already close to our initial goal to formulate the equations
as an evo-system. Here $\interior\curl$ denotes the  $L^{2}$ -closure
of the classical $\curl$ defined on $C_{1}\left(\mathbb{R}^{3}\right)$-vector
fields vanishing outside closed, bounded subsets of $\mathbb{R}^{3}$.
Moreover, $\curl\coloneqq\interior\curl^{*}$and so the spatial Maxwell
operator is skew-selfadjoint in $L^{2}\left(\mathbb{R}^{3},\mathbb{R}^{6}\right)$.
In case of a domain $\Omega$ with boundary we take $\interior\curl$
constructed analogously with $C_{1}\left(\Omega\right)$-vector fields
vanishing outside closed, bounded sets contained in $\Omega$, where
$\Omega$ is a non-empty open set in $\mathbb{R}^{3}$ (\textbf{strong}
definition of $\interior\curl$) and define
\begin{equation}
\curl\coloneqq\interior\curl^{*}\label{eq:weak-1}
\end{equation}

(\textbf{weak}\footnote{\textbf{\label{fn:Of-course-weak}}Of course\textbf{ }``weak equals
strong''. It is $C_{1}\left(\Omega\right)\cap D\left(\curl\right)$
dense in $D\left(\curl\right)$ by T. Kasuga's argument, see \cite{Ref275},
\cite[section 2.1]{Leis:Buch:2}, the strong definition of $\curl$
as the closure $\overline{\curl\restricted{C_{1}\left(\Omega\right)\cap D\left(\curl\right)}}$
equals its weak definition. Consequently, also $\interior\curl=\curl^{*}=\left(\curl\restricted{C_{1}\left(\Omega\right)\cap D\left(\curl\right)}\right)^{*}$,
which confirms ``weak equals strong'' for $\interior\curl$ as well. } definition of $\curl$). Thus we arrive indeed at the evo-system\footnote{Here we have thrown in an extra magnetic external source term, since
mathematically it is no obstacle to treat $k_{\mathrm{ext}}\not=0$.} 
\[
\left(\partial_{0}M\left(\partial_{0}^{-1}\right)+A\right)\left(\begin{array}{c}
E\\
H
\end{array}\right)=\left(\begin{array}{c}
-j_{\mathrm{ext}}\\
k_{\mathrm{ext}}
\end{array}\right)
\]
{\small{}with $M\left(\partial_{0}^{-1}\right)=M\left(0\right)+\partial_{0}^{-1}M^{\prime}\left(0\right)$
and here specifically
\begin{equation}
M\left(0\right)=\left(\begin{array}{cc}
\epsilon & 0\\
0 & \mu
\end{array}\right),\,M^{\prime}\left(0\right)=\left(\begin{array}{cc}
\sigma & 0\\
0 & 0
\end{array}\right),\,A=\left(\begin{array}{cc}
0 & -\curl\\
\interior\curl & 0
\end{array}\right).\label{eq:em-material}
\end{equation}
}which satisfies the well-posedness constraint if we assume $\epsilon,\mu$
selfadjoint and (compare \eqref{eq:symdef} and \eqref{fn:sym})
\begin{equation}
\nu\epsilon+\mathrm{sym}\left(\sigma\right),\mu\geq\eta_{0}>0,\label{eq:posdef-em}
\end{equation}
for all sufficiently large $\nu\in\oi0\infty$. Note that with this
assumption also $\epsilon$ having a non-trivial null space, the so-called
eddy current problem, can be handled without further adjustments.
Of course, in the spirit of Theorem \ref{thm:well} we could consider
more general media. More recently, so-called electro-magnetic \emph{metamaterials}
have come into focus, which are media, where $M^{\prime\prime}\not=0$
or $M\left(z\right)$ is \emph{not} block-diagonal. To classify some
prominent cases, there are for example:
\begin{itemize}
\item Bi-anisotropic media, characterized by 
\[
M\left(0\right)=\left(\begin{array}{cc}
\epsilon & \kappa^{*}\\
\kappa & \mu
\end{array}\right),\:\kappa\not=0.
\]
Since, due to \eqref{eq:posdef}, we must have $M\left(0\right)\geq0$,
we get $\epsilon\geq0$ and 
\[
\left|\mu^{-1/2}\kappa\epsilon^{-1/2}\right|\leq1.
\]
Note that this is a strong smallness constraint on the off-diagonal
entry $\kappa$. For example in homogeneous, isotropic media $c_{0}=\epsilon^{-1/2}\mu^{-1/2}$
is the speed of light and the above condition yields
\[
\left|\kappa\right|\leq\frac{1}{c_{0}}.
\]
\item Chiral media: 
\[
M^{\prime}\left(0\right)=\left(\begin{array}{cc}
0 & -\chi\\
\chi & 0
\end{array}\right),\:\chi\not=0\mbox{ selfadjoint}.
\]
\item Omega media: 
\[
M^{\prime}\left(0\right)=\left(\begin{array}{cc}
0 & \chi\\
\chi & 0
\end{array}\right),\:\chi\not=0\mbox{ skew-selfadjoint}.
\]
\end{itemize}

\subsection{The Equations of Linear Elasto-Dynamics as an Evo-System }

Linear elasto-dynamics is usually discussed in a symmetric tensor-valued
$L^{2}$-setting for the stress $T$, i.e. $T\in L^{2}\left(\Omega,\mathrm{sym}\left[\mathbb{R}^{3\times3}\right]\right)$,
and a vector $L^{2}$-setting for the displacement $u\in L^{2}\left(\Omega,\mathbb{R}^{3}\right)$.
Here $\mathrm{sym}$ is the (orthogonal) projector onto real-symmetric-matrix-valued
$L^{2}$-functions. More precisely, we extend $\mathrm{sym}$ to the
matrix-valued case by letting
\begin{eqnarray*}
\mathrm{sym}:L^{2}\left(\Omega,\mathbb{R}^{3\times3}\right) & \to & L^{2}\left(\Omega,\mathbb{R}^{3\times3}\right),\\
W & \mapsto & \frac{1}{2}\left(W+W^{*}\right),
\end{eqnarray*}
where the adjoint $W^{*}$ is taken point-wise by the standard Frobenius
inner product 
\[
\left(T,S\right)\mapsto\mathrm{trace}\left(T^{\top}S\right)
\]
for $3\times3$-matrices, such that 
\begin{eqnarray*}
\mathbb{R}^{3\times3} & \to & \mathbb{R}^{6}\\
\left(\begin{array}{ccc}
T_{00} & T_{01} & T_{02}\\
T_{10} & T_{11} & T_{12}\\
T_{20} & T_{21} & T_{22}
\end{array}\right) & \mapsto & \left(\begin{array}{c}
T_{00}\\
T_{11}\\
T_{22}\\
T_{12}\\
T_{20}\\
T_{01}\\
T_{21}\\
T_{02}\\
T_{10}
\end{array}\right)
\end{eqnarray*}
is unitary. Then with
\begin{eqnarray*}
\iota_{\mathrm{sym}}:L^{2}\left(\Omega,\mathrm{sym}\left[\mathbb{R}^{3\times3}\right]\right) & \to & L^{2}\left(\Omega,\mathbb{R}^{3\times3}\right)\\
T & \mapsto & T,
\end{eqnarray*}
denoting the canonical embedding of the subspace $L^{2}\left(\Omega,\mathrm{sym}\left[\mathbb{R}^{3\times3}\right]\right)$
in $L^{2}\left(\Omega,\mathbb{R}^{3\times3}\right)$ we have
\begin{eqnarray*}
\iota_{\mathrm{sym}}^{*}:L^{2}\left(\Omega,\mathbb{R}^{3\times3}\right) & \to & L^{2}\left(\Omega,\mathrm{sym}\left[\mathbb{R}^{3\times3}\right]\right)\\
W & \mapsto & \mathrm{sym}W
\end{eqnarray*}
and so we have the useful factorization
\[
\mathrm{sym}=\iota_{\mathrm{sym}}\iota_{\mathrm{sym}}^{*}.
\]
With this observation we can now approach the standard equations of
elasticity theory. The dynamics of elastic processes is commonly captured
in a second order formulation for the displacement $u$ by
\[
\rho_{*}\partial_{0}^{2}u-\Div C\Grad u=f,
\]
where 
\begin{eqnarray*}
\Grad u & \coloneqq & \iota_{\mathrm{sym}}^{*}\left(\nabla u\right)\\
\Div T & \coloneqq & \left(\nabla^{\top}T\right)^{\top}
\end{eqnarray*}
for symmetric $T$, i.e. $T\in L^{2}\left(\Omega,\mathrm{sym}\left[\mathbb{R}^{3\times3}\right]\right)$.
The elasticity `tensor', i.e. rather the mapping
\[
C:L^{2}\left(\Omega,\mathrm{sym}\left[\mathbb{R}^{3\times3}\right]\right)\to L^{2}\left(\Omega,\mathrm{sym}\left[\mathbb{R}^{3\times3}\right]\right)
\]
and the mass density operator
\[
\rho_{*}:L^{2}\left(\Omega,\mathbb{R}^{3}\right)\to L^{2}\left(\Omega,\mathbb{R}^{3}\right)
\]
are assumed to be selfadjoint and strictly positive definite.

The origin, from which the above second order system is derived, is
naturally a system of algebraic and first order differential equations.
The original system can be easily reconstructed by re-introducing
the relevant physical quantities velocity $v\coloneqq\partial_{0}u$
and stress $T\coloneqq C\Grad u$. Thus, we arrive at the system
\begin{eqnarray*}
\rho_{*}\partial_{0}v-\Div T & = & f,\\
T & = & C\Grad\partial_{0}^{-1}v,
\end{eqnarray*}
in the unknowns $v$ and $T$. Differentiating the second equation
with respect to time, we end up with a system of the block operator
matrix form
\begin{align*}
\left(\partial_{0}\left(\begin{array}{cc}
\rho_{*} & 0\\
0 & C^{-1}
\end{array}\right)+\left(\begin{array}{cc}
0 & -\Div\\
-\Grad & 0
\end{array}\right)\right)\left(\begin{array}{c}
v\\
T
\end{array}\right) & =\left(\begin{array}{c}
f\\
0
\end{array}\right).
\end{align*}
Choosing now for example a homogeneous Dirichlet boundary condition,
i.e. we replace $\Grad$ by\footnote{Korn's inequality shows that the closure bar is superfluous 
\[
\interior\Grad=\iota_{\mathrm{sym}}^{*}\interior\grad.
\]
} 
\[
\interior\Grad\coloneqq\overline{\iota_{\mathrm{sym}}^{*}\interior\grad},
\]
where $\interior\grad$ is the closure of differentiation for vector
fields (the Jacobian matrix) with compact support in $\Omega$ as
a mapping from $L^{2}\left(\Omega,\mathbb{R}^{3}\right)$ to $L^{2}\left(\Omega,\mathbb{R}^{3\times3}\right)$,
and 
\[
\Div\coloneqq\dive\iota_{\mathrm{sym}}
\]
so that
\[
\interior\Grad=-\Div^{*},
\]
we are led to consider an evo-system of the form

\begin{align}
\left(\partial_{0}\left(\begin{array}{cc}
\rho_{*} & 0\\
0 & C^{-1}
\end{array}\right)+\left(\begin{array}{cc}
0 & -\Div\\
-\interior\Grad & 0
\end{array}\right)\right)\left(\begin{array}{c}
v\\
T
\end{array}\right) & =\left(\begin{array}{c}
f\\
g
\end{array}\right).\label{eq:elast}
\end{align}
\begin{rem}We note that also here we have ``weak equals strong''
following the same rationale as in the electro-magneto-dynamics case,
compare Footnote \ref{fn:Of-course-weak}.\end{rem} In the light
of \eqref{eq:posdef} the well-posedness results from assuming that
\begin{equation}
\rho_{*},C\geq\eta_{0}>0\label{eq:posdef-elast}
\end{equation}
for some real constant $\eta_{0}$.

\section{\label{sec:An-Interface-Coupling}An Interface Coupling Mechanism.}

After the above preliminary considerations, we are now ready to consider
the situation, where the electro-magnetic field in one region interacts
with elastic media in another region via some common interface. Rather
than basing our choice of transmission constraints on the interface
by physical arguments, we shall explore a deep connection between
electro-magneto-dynamics and elasto-dynamics to arrive at natural
transmission conditions built into the construction of the evo-system.
This construction will utilize the idea of a \emph{mother-descendant}
construction introduced in \cite{zbMATH06250993}, see \cite{ZAMM:ZAMM201300297}
for a more viable version, which we will briefly recall.

\subsection{Mother Operators and their Descendants}

We recall from \cite{zbMATH06250993} the following simple but crucial
lemma.

\begin{lem}Let $C:D\left(C\right)\subseteq H\to Y$ be a closed densely-defined
linear operator between Hilbert spaces $H,\:Y$. Moreover, let $B:Y\to X$
be a continuous linear operator into another Hilbert space $X$. If
$C^{*}B^{*}$ is densely defined, then
\[
\overline{BC}=\left(C^{*}B^{*}\right)^{*}.
\]
\end{lem}\begin{proof}It is 
\[
C^{*}B^{*}\subseteq\left(BC\right)^{*}.
\]
If $\phi\in D\left(\left(BC\right)^{*}\right)$ then 
\[
\left\langle BCu|\phi\right\rangle _{X}=\left\langle u|\left(BC\right)^{*}\phi\right\rangle _{H}
\]
for all $u\in D\left(C\right)$. Thus, we have
\[
\left\langle Cu|B^{*}\phi\right\rangle _{Y}=\left\langle BCu|\phi\right\rangle _{X}=\left\langle u|\left(BC\right)^{*}\phi\right\rangle _{H}
\]
for all $u\in D\left(C\right)$ and we read off that $B^{*}\phi\in D\left(C^{*}\right)$
and 
\[
C^{*}B^{*}\phi=\left(BC\right)^{*}\phi.
\]
Thus we have
\[
\left(BC\right)^{*}=C^{*}B^{*}.
\]
If now $C^{*}B^{*}$ is densely defined, we have for its adjoint operator
\[
\left(C^{*}B^{*}\right)^{*}=\overline{BC}.
\]
\end{proof} As a consequence we have that the \emph{descendant}
\[
\overline{\left(\begin{array}{cc}
1 & 0\\
0 & B
\end{array}\right)\left(\begin{array}{cc}
0 & -C^{*}\\
C & 0
\end{array}\right)}\left(\begin{array}{cc}
1 & 0\\
0 & B^{*}
\end{array}\right)=\left(\begin{array}{cc}
0 & -C^{*}B^{*}\\
\overline{BC} & 0
\end{array}\right)
\]
indeed inherits its skew-selfadjointness from its \emph{mother} $\left(\begin{array}{cc}
0 & -C^{*}\\
C & 0
\end{array}\right)$ (with $C$ replaced by $\overline{BC}$). Moreover, we record the
following result on the stability of well-posedness in the mother-descendant
process.

\begin{thm}\label{th:mother}Let $C:D\left(C\right)\subseteq H\to Y$
be a closed densely-defined linear operator between Hilbert spaces
$H,\:Y$. Moreover, let $B:Y\to X$ be a continuous linear operator
into another Hilbert space $X$ with a closed range $B\left[Y\right]$
such that $C^{*}B^{*}$ is densely defined. Then, if 
\[
\left(\partial_{0}M\left(\partial_{0}^{-1}\right)+\left(\begin{array}{cc}
0 & -C^{*}\\
C & 0
\end{array}\right)\right)\left(\begin{array}{c}
U_{0}\\
U_{1}
\end{array}\right)=\left(\begin{array}{c}
F_{0}\\
F_{1}
\end{array}\right)
\]
with data $\left(\begin{array}{c}
F_{0}\\
F_{1}
\end{array}\right)\in H_{\nu}\left(\mathbb{R},H\oplus X\right)$ and a solution $\left(\begin{array}{c}
U_{0}\\
U_{1}
\end{array}\right)\in H_{\nu}\left(\mathbb{R},H\oplus X\right)$ is a well-posed evo-system (satisfying in particular \eqref{eq:posdef}),
so is the descendant problem
\[
\left(\partial_{0}\tilde{M}\left(\partial_{0}^{-1}\right)+\tilde{A}\right)U=\left(\begin{array}{c}
F_{0}\\
G_{1}
\end{array}\right)\in H_{\nu}\left(\mathbb{R},H\oplus B\left[Y\right]\right),
\]
where
\begin{eqnarray*}
\tilde{M}\left(\partial_{0}^{-1}\right) & = & \left(\begin{array}{cc}
1 & 0\\
0 & B
\end{array}\right)M\left(\partial_{0}^{-1}\right)\left(\begin{array}{cc}
1 & 0\\
0 & B^{*}
\end{array}\right),\\
\tilde{A} & = & \overline{\left(\begin{array}{cc}
1 & 0\\
0 & B
\end{array}\right)\left(\begin{array}{cc}
0 & -C^{*}\\
C & 0
\end{array}\right)}\left(\begin{array}{cc}
1 & 0\\
0 & B^{*}
\end{array}\right).
\end{eqnarray*}

\end{thm}

\begin{proof}The positive-definiteness condition \eqref{eq:posdef}
carries over to the new material law operator in the following way.
If
\[
\nu M\left(0\right)+\mathrm{sym}\left(M^{\prime}\left(0\right)\right)\geq c_{*}>0
\]
for all $\nu\in\lci{\nu_{0}}\infty$ and some $\nu_{0}\in\oi0\infty$
, then
\begin{eqnarray*}
\nu\tilde{M}\left(0\right)+\mathrm{sym}\left(\tilde{M}^{\prime}\left(0\right)\right) & = & \nu\left(\begin{array}{cc}
1 & 0\\
0 & B
\end{array}\right)M\left(0\right)\left(\begin{array}{cc}
1 & 0\\
0 & B^{*}
\end{array}\right)+\\
 &  & +\mathrm{sym}\left(\left(\begin{array}{cc}
1 & 0\\
0 & B
\end{array}\right)M^{\prime}\left(0\right)\left(\begin{array}{cc}
1 & 0\\
0 & B^{*}
\end{array}\right)\right)
\end{eqnarray*}
and we estimate for $\left(V_{0},V_{1}\right)\in H\oplus B\left[Y\right]$

\begin{eqnarray*}
\nu\left\langle \left(\begin{array}{c}
V_{0}\\
V_{1}
\end{array}\right)\Big|\left(\begin{array}{cc}
1 & 0\\
0 & B
\end{array}\right)M\left(0\right)\left(\begin{array}{cc}
1 & 0\\
0 & B^{*}
\end{array}\right)\left(\begin{array}{c}
V_{0}\\
V_{1}
\end{array}\right)\right\rangle _{H\oplus B\left[Y\right]}+\\
+\left\langle \left(\begin{array}{c}
V_{0}\\
V_{1}
\end{array}\right)\Big|\left(\begin{array}{cc}
1 & 0\\
0 & B
\end{array}\right)\mathrm{sym}\left(M^{\prime}\left(0\right)\right)\left(\begin{array}{cc}
1 & 0\\
0 & B^{*}
\end{array}\right)\left(\begin{array}{c}
V_{0}\\
V_{1}
\end{array}\right)\right\rangle _{H\oplus B\left[Y\right]}=\\
=\nu\left\langle \left(\begin{array}{cc}
1 & 0\\
0 & B^{*}
\end{array}\right)\left(\begin{array}{c}
V_{0}\\
V_{1}
\end{array}\right)\Big|M\left(0\right)\left(\begin{array}{cc}
1 & 0\\
0 & B^{*}
\end{array}\right)\left(\begin{array}{c}
V_{0}\\
V_{1}
\end{array}\right)\right\rangle _{H\oplus Y}+
\end{eqnarray*}
\begin{eqnarray*}
\\
+\left\langle \left(\begin{array}{cc}
1 & 0\\
0 & B^{*}
\end{array}\right)\left(\begin{array}{c}
V_{0}\\
V_{1}
\end{array}\right)\Big|\mathrm{sym}\left(M^{\prime}\left(0\right)\right)\left(\begin{array}{cc}
1 & 0\\
0 & B^{*}
\end{array}\right)\left(\begin{array}{c}
V_{0}\\
V_{1}
\end{array}\right)\right\rangle _{H\oplus Y},\\
\geq c_{*}\left\langle \left(\begin{array}{cc}
1 & 0\\
0 & B^{*}
\end{array}\right)\left(\begin{array}{c}
V_{0}\\
V_{1}
\end{array}\right)\Big|\left(\begin{array}{cc}
1 & 0\\
0 & B^{*}
\end{array}\right)\left(\begin{array}{c}
V_{0}\\
V_{1}
\end{array}\right)\right\rangle _{H\oplus Y}\\
\geq\tilde{c}_{*}\left\langle \left(\begin{array}{c}
V_{0}\\
V_{1}
\end{array}\right)\Big|\left(\begin{array}{c}
V_{0}\\
V_{1}
\end{array}\right)\right\rangle _{H\oplus B\left[Y\right]}
\end{eqnarray*}
Indeed, since by the closed range asumption $B\left[Y\right]$ and
$B^{*}\left[X\right]$ are Hilbert spaces and by the closed graph
theorem the operator
\begin{eqnarray*}
\left(\begin{array}{cc}
1 & 0\\
0 & B^{*}\iota_{B\left[Y\right]}
\end{array}\right):H\oplus B\left[Y\right] & \to & H\oplus B^{*}\left[X\right]\\
\left(\begin{array}{c}
V_{0}\\
V_{1}
\end{array}\right) & \mapsto & \left(\begin{array}{c}
V_{0}\\
B^{*}V_{1}
\end{array}\right)
\end{eqnarray*}
has a continuous inverse, we have
\begin{eqnarray*}
 &  & \left|\left(\begin{array}{c}
V_{0}\\
V_{1}
\end{array}\right)\right|_{H\oplus B\left[Y\right]}=\\
 & = & \left|\left(\begin{array}{cc}
1 & 0\\
0 & B^{*}\iota_{B\left[Y\right]}
\end{array}\right)^{-1}\left(\begin{array}{cc}
1 & 0\\
0 & B^{*}
\end{array}\right)\left(\begin{array}{c}
V_{0}\\
V_{1}
\end{array}\right)\right|_{H\oplus B\left[Y\right]}\\
 & \leq & \left\Vert \left(\begin{array}{cc}
1 & 0\\
0 & B^{*}\iota_{B\left[Y\right]}
\end{array}\right)^{-1}\right\Vert \left|\left(\begin{array}{cc}
1 & 0\\
0 & B^{*}
\end{array}\right)\left(\begin{array}{c}
V_{0}\\
V_{1}
\end{array}\right)\right|_{H\oplus Y}
\end{eqnarray*}
and so we may choose
\[
\tilde{c}_{*}=c_{*}\left\Vert \left(\begin{array}{cc}
1 & 0\\
0 & B^{*}\iota_{B\left[Y\right]}
\end{array}\right)^{-1}\right\Vert ^{-2}
\]
to confirm that
\[
\nu\tilde{M}\left(0\right)+\mathrm{sym}\left(\tilde{M}^{\prime}\left(0\right)\right)\geq\tilde{c}_{*}>0
\]
for all $\nu\in\lci{\nu_{0}}\infty$ and some $\nu_{0}\in\oi0\infty$~.\end{proof}

As a particular instance of this construction we can take $B$ specifically
as $\iota_{S}^{*}$, where $\iota_{S}:S\to H$, $x\mapsto x$, is
the canonical embedding of the closed subspace $S$ in $H$. Then
\[
\overline{\left(\begin{array}{cc}
1 & 0\\
0 & \iota_{S}^{*}
\end{array}\right)\left(\begin{array}{cc}
0 & -C\\
C^{*} & 0
\end{array}\right)}\left(\begin{array}{cc}
1 & 0\\
0 & \iota_{S}
\end{array}\right)=\left(\begin{array}{cc}
0 & -C\iota_{S}\\
\overline{\iota_{S}^{*}C^{*}} & 0
\end{array}\right)
\]
is skew-selfadjoint if $C\:\iota_{S}:D\left(C\right)\cap S\subseteq S\to Y$,
the restriction of $C:D\left(C\right)\subseteq H\to Y$ to the closed
subspace $S\subseteq H$ is densely defined in $S$. This is the construction
we shall employ to approach our specific problem. First we observe
that both physical regimes do indeed have the same \emph{mother}.

\subsection{\label{subsec:Two-Descendants-of}Two Descendants of Non-symmetric
Elasticity }

As a convenient mother to start from we take the theory of non-symmetric
elasticity, W. Nowacki, \cite{zbMATH03315043,Nowacki1986}, leading
to an evo-system of the form
\[
\left(\partial_{0}M_{0}+M_{1}+\left(\begin{array}{cc}
0 & -\dive\\
-\interior\grad & 0
\end{array}\right)\right)\left(\begin{array}{c}
v\\
T
\end{array}\right)=\left(\begin{array}{c}
f\\
g
\end{array}\right).
\]

We shall now discuss two particular descendants.
\begin{enumerate}
\item Classical symmetric elasticity theory can be considered as a descendant
of the form
\begin{align*}
 & \left(\partial_{0}\left(\begin{array}{cc}
1 & 0\\
0 & \iota_{\mathrm{sym}}^{*}
\end{array}\right)M_{0}\left(\begin{array}{cc}
1 & 0\\
0 & \iota_{\mathrm{sym}}
\end{array}\right)+\left(\begin{array}{cc}
1 & 0\\
0 & \iota_{\mathrm{sym}}^{*}
\end{array}\right)M_{1}\left(\begin{array}{cc}
1 & 0\\
0 & \iota_{\mathrm{sym}}
\end{array}\right)+\right.\\
 & \left.+\left(\begin{array}{cc}
0 & -\Div\\
-\interior\Grad & 0
\end{array}\right)\right)\left(\begin{array}{c}
v\\
T_{\mathrm{sym}}
\end{array}\right)=\left(\begin{array}{c}
f\\
g_{\mathrm{sym}}
\end{array}\right),
\end{align*}
where 
\[
\interior\Grad\coloneqq\overline{\iota_{\mathrm{sym}}^{*}\interior\grad}
\]
and 
\[
\Div\coloneqq\dive\iota_{\mathrm{sym}}.
\]
Note that the assumptions of Theorem \ref{th:mother} are clearly
satisfied since smooth elements with compact support are already a
dense sub-domain of $\dive\iota_{\mathrm{sym}}$. In the classical
situation, which we shall assume for simplicity, we have $M_{1}=0$
and
\[
M_{0}=\left(\begin{array}{cc}
\rho_{*} & 0\\
0 & C^{-1}
\end{array}\right).
\]
\item Maxwell's equation are obtained in a sense by the opposite construction.\\
If we denote analogously
\begin{eqnarray*}
\mathrm{skew}:L^{2}\left(\Omega,\mathbb{R}^{3\times3}\right) & \to & L^{2}\left(\Omega,\mathbb{R}^{3\times3}\right),\\
W & \mapsto & \frac{1}{2}\left(W-W^{*}\right),
\end{eqnarray*}
then with
\begin{eqnarray*}
\iota_{\mathrm{skew}}:L^{2}\left(\Omega,\mathrm{skew}\left[\mathbb{R}^{3\times3}\right]\right) & \to & L^{2}\left(\Omega,\mathbb{R}^{3\times3}\right)\\
T & \mapsto & T,
\end{eqnarray*}
denoting the canonical embedding of $L^{2}\left(\Omega,\mathrm{skew}\left[\mathbb{R}^{3\times3}\right]\right)$
in $L^{2}\left(\Omega,\mathbb{R}^{3\times3}\right)$ we find
\begin{eqnarray*}
\iota_{\mathrm{skew}}^{*}:L^{2}\left(\Omega,\mathbb{R}^{3\times3}\right) & \to & L^{2}\left(\Omega,\mathrm{skew}\left[\mathbb{R}^{3\times3}\right]\right)\\
W & \mapsto & \mathrm{skew}W.
\end{eqnarray*}
With this we may now construct the Maxwell evo-system as
\begin{align*}
 & \left(\partial_{0}\left(\begin{array}{cc}
1 & 0\\
0 & -\sqrt{2}I^{*}\iota_{\mathrm{skew}}^{*}
\end{array}\right)M_{0}\left(\begin{array}{cc}
1 & 0\\
0 & -\sqrt{2}\iota_{\mathrm{skew}}I
\end{array}\right)+\right.\\
 & +\left(\begin{array}{cc}
1 & 0\\
0 & -\sqrt{2}I^{*}\iota_{\mathrm{skew}}^{*}
\end{array}\right)M_{1}\left(\begin{array}{cc}
1 & 0\\
0 & -\sqrt{2}\iota_{\mathrm{skew}}I
\end{array}\right)+\\
 & \left.+\left(\begin{array}{cc}
0 & -\mathrm{curl}\\
\interior{\mathrm{curl}} & 0
\end{array}\right)\right)\left(\begin{array}{c}
E\\
H
\end{array}\right)=\left(\begin{array}{c}
f\\
-I^{*}g_{\mathrm{skew}}
\end{array}\right),
\end{align*}
where 
\[
I:\left(\begin{array}{c}
\alpha_{1}\\
\alpha_{2}\\
\alpha_{3}
\end{array}\right)\mapsto\frac{1}{\sqrt{2}}\left(\begin{array}{ccc}
0 & -\alpha_{3} & \alpha_{2}\\
\alpha_{3} & 0 & -\alpha_{1}\\
-\alpha_{2} & \alpha_{1} & 0
\end{array}\right)
\]
 is a unitary transformation and so is its inverse 
\[
I^{*}:\left(\begin{array}{ccc}
0 & -\alpha_{3} & \alpha_{2}\\
\alpha_{3} & 0 & -\alpha_{1}\\
-\alpha_{2} & \alpha_{1} & 0
\end{array}\right)\mapsto\sqrt{2}\left(\begin{array}{c}
\alpha_{1}\\
\alpha_{2}\\
\alpha_{3}
\end{array}\right).
\]
Again, for simplicity we focus on the classical choice of \eqref{eq:em-material}.
We calculate 
\begin{eqnarray*}
 &  & I^{*}\iota_{\mathrm{skew}}^{*}\grad v=\\
 & = & \frac{1}{2}I^{*}\left(\begin{array}{ccc}
0 & \partial_{2}v_{1}-\partial_{1}v_{2} & \partial_{3}v_{1}-\partial_{1}v_{3}\\
\partial_{1}v_{2}-\partial_{2}v_{1} & 0 & \partial_{3}v_{2}-\partial_{2}v_{3}\\
\partial_{1}v_{3}-\partial_{3}v_{1} & \partial_{2}v_{3}-\partial_{3}v_{2} & 0
\end{array}\right)\\
 & = & -\frac{1}{\sqrt{2}}\left(\begin{array}{c}
\partial_{3}v_{2}-\partial_{2}v_{3}\\
\partial_{1}v_{3}-\partial_{3}v_{1}\\
\partial_{2}v_{1}-\partial_{1}v_{2}
\end{array}\right)=\frac{1}{\sqrt{2}}\left(\begin{array}{c}
\partial_{2}v_{3}-\partial_{3}v_{2}\\
\partial_{3}v_{1}-\partial_{1}v_{3}\\
\partial_{1}v_{2}-\partial_{2}v_{1}
\end{array}\right)\\
 & \eqqcolon & \frac{1}{\sqrt{2}}\curl v
\end{eqnarray*}
and also confirm that 
\[
\dive\;\iota_{\mathrm{skew}}I=-\frac{1}{\sqrt{2}}\curl.
\]
In other terms, we have the congruence to a descendant 
\begin{eqnarray*}
\left(\begin{array}{cc}
0 & -\curl\\
\interior\curl & 0
\end{array}\right)=\\
=\left(\begin{array}{cc}
1 & 0\\
0 & -\sqrt{2}I^{*}
\end{array}\right)\left(\begin{array}{cc}
0 & -\dive\:\iota_{\mathrm{skew}}\\
-\overline{\iota_{\mathrm{skew}}^{*}\interior\grad} & 0
\end{array}\right)\left(\begin{array}{cc}
1 & 0\\
0 & -\sqrt{2}I
\end{array}\right),
\end{eqnarray*}
where we have used that
\end{enumerate}
\[
\interior\curl=\sqrt{2}\:I^{*}\overline{\iota_{\mathrm{skew}}^{*}\interior\grad}.
\]
Note that again smooth elements with compact support are a dense sub-domain
of $\dive\:\iota_{\mathrm{skew}}$ and so the assumptions of Theorem
\ref{th:mother} are clearly satisfied. Motivated by the observation
that Maxwell's equations and the (symmetric) elasto-dynamic equations
are both descendants from the asymmetric elasto-dynamics equations
of Novacki, \cite{zbMATH03315043,Nowacki1986}, we will now discuss
boundary interactions between both systems.

\subsection{An Application to Interface Coupling}

Motivated by a paper of F. Cakoni \& G.C. Hsiao, \cite{zbMATH02114868},
where the time-harmonic isotropic homogeneous case of electro-dynamics
and elasticity, respectively, is studied via transmission conditions
across a separating interface, we consider the corresponding time-dependent
case. We assume $\Omega_{0}\cup\Omega_{1}\subseteq\Omega$, such that
the orthogonal decompositions 
\begin{eqnarray}
L^{2}\left(\Omega,\mathbb{R}^{3\times3}\right) & = & L^{2}\left(\Omega_{0},\mathbb{R}^{3\times3}\right)\oplus L^{2}\left(\Omega_{1},\mathbb{R}^{3\times3}\right)\nonumber \\
L^{2}\left(\Omega,\mathbb{R}^{3}\right) & = & L^{2}\left(\Omega_{0},\mathbb{R}^{3}\right)\oplus L^{2}\left(\Omega_{1},\mathbb{R}^{3}\right)\label{eq:deco-v}
\end{eqnarray}
hold, and let $I_{0}\coloneqq\left(\begin{array}{cc}
\iota_{L^{2}\left(\Omega_{0},\mathrm{sym}\left[\mathbb{R}^{3\times3}\right]\right)} & -\iota_{L^{2}\left(\Omega_{1},\mathrm{skew}\left[\mathbb{R}^{3\times3}\right]\right)}\sqrt{2}I\end{array}\right)$, i.e. 
\[
I_{0}\left(\begin{array}{c}
S\\
v
\end{array}\right)=\iota_{L^{2}\left(\Omega_{0},\mathrm{sym}\left[\mathbb{R}^{3\times3}\right]\right)}S-\iota_{L^{2}\left(\Omega_{1},\mathrm{skew}\left[\mathbb{R}^{3\times3}\right]\right)}\sqrt{2}Iv
\]
with the respective canonical embeddings into $L^{2}\left(\Omega,\mathbb{R}^{3\times3}\right)$.
Then 
\begin{eqnarray*}
I_{0}^{*}:L^{2}\left(\Omega,\mathbb{R}^{3\times3}\right) & \to & L^{2}\left(\Omega_{0},\mathrm{sym}\left[\mathbb{R}^{3\times3}\right]\right)\oplus L^{2}\left(\Omega_{1},\mathbb{R}^{3}\right),\\
T & \mapsto & \left(\begin{array}{c}
\iota_{L^{2}\left(\Omega_{0},\mathrm{sym}\left[\mathbb{R}^{3\times3}\right]\right)}^{*}T\\
-\sqrt{2}I^{*}\iota_{L^{2}\left(\Omega_{1},\mathrm{skew}\left[\mathbb{R}^{3\times3}\right]\right)}^{*}T
\end{array}\right),
\end{eqnarray*}
and so
\[
I_{0}^{*}=\left(\begin{array}{c}
\iota_{L^{2}\left(\Omega_{0},\mathrm{sym}\left[\mathbb{R}^{3\times3}\right]\right)}^{*}\\
-\sqrt{2}I^{*}\iota_{L^{2}\left(\Omega_{1},\mathrm{skew}\left[\mathbb{R}^{3\times3}\right]\right)}^{*}
\end{array}\right).
\]

With  this we get a congruence to a descendant construction as
\begin{eqnarray}
A & = & \overline{\left(\begin{array}{cc}
1 & 0\\
0 & I_{0}^{*}
\end{array}\right)\left(\begin{array}{cc}
0 & -\dive\\
-\interior\grad & 0
\end{array}\right)}\left(\begin{array}{cc}
1 & 0\\
0 & I_{0}
\end{array}\right)\label{eq:incl}\\
 &  & \subseteq\left(\begin{array}{cc}
0 & \left(\begin{array}{cc}
-\Div_{\Omega_{0}} & -\curl_{\Omega_{1}}\end{array}\right)\\
\left(\begin{array}{c}
-\Grad_{\Omega_{0}}\\
\curl_{\Omega_{1}}
\end{array}\right) & \left(\begin{array}{cc}
0 & 0\\
0 & 0
\end{array}\right)
\end{array}\right)\label{eq:incl-1}
\end{eqnarray}
and 
\begin{align}
M\left(0\right) & =\left(\begin{array}{cc}
\rho_{*,\Omega_{0}}+\epsilon_{\Omega_{1}} & \left(\begin{array}{cc}
0 & 0\end{array}\right)\\
\left(\begin{array}{c}
0\\
0
\end{array}\right) & \left(\begin{array}{cc}
C_{\Omega_{0}}^{-1} & 0\\
0 & \mu_{\Omega_{1}}
\end{array}\right)
\end{array}\right)\label{eq:M0-mx}\\
M^{\prime}\left(0\right) & =\left(\begin{array}{cc}
\sigma_{\Omega_{1}} & \left(\begin{array}{cc}
0 & 0\end{array}\right)\\
\left(\begin{array}{c}
0\\
0
\end{array}\right) & \left(\begin{array}{cc}
0 & 0\\
0 & 0
\end{array}\right)
\end{array}\right).\label{eq:M1-mx}
\end{align}
The indexes $\Omega_{k}$, $k=0,1,$ are used to denote the respective
supports of the quantities. The coefficients are – as a matter of
simplification labeled in the same meaning as in \eqref{eq:em-material}
and \eqref{eq:elast}, just with the support information added\footnote{Although we consider for convenience and physical relevance this evo-system
in its own right, a formal mother material law – without physical
meaning – could be easily given:
\[
\left(\begin{array}{cc}
\rho_{*,\Omega_{0}}+\epsilon_{\Omega_{1}}+\partial_{0}^{-1}\sigma_{\Omega_{1}} & 0\\
0 & m_{11}
\end{array}\right)
\]
with for example
\[
m_{11}=\iota_{\mathrm{sym},\Omega_{0}}^{*}C_{\Omega_{0}}^{-1}\iota_{\mathrm{sym},\Omega_{0}}+\iota_{\mathrm{skew},\Omega_{0}}^{*}\iota_{\mathrm{skew},\Omega_{0}}+\iota_{\mathrm{skew},\Omega_{1}}^{*}\mu_{\Omega_{1}}\iota_{\mathrm{skew},\Omega_{1}}+\iota_{\mathrm{sym},\Omega_{1}}^{*}\iota_{\mathrm{sym},\Omega_{1}}.
\]
Then the described mother-descendant mechanism would lead to a descendant,
which in turn would be congruent to the described interface system.}. The unknowns are now 
\[
\left(\begin{array}{c}
v_{\Omega_{0}}+E_{\Omega_{1}}\\
\left(\begin{array}{c}
T_{\Omega_{0}}\\
H_{\Omega_{1}}
\end{array}\right)
\end{array}\right)\in H=L^{2}\left(\Omega,\mathbb{R}^{3}\right)\oplus\left(L^{2}\left(\Omega_{0},\mathrm{sym}\left[\mathbb{R}^{3\times3}\right]\right)\oplus L^{2}\left(\Omega_{1},\mathbb{R}^{3}\right)\right),
\]
where the first component is to be understood in the sense of \eqref{eq:deco-v}.
Note that the assumptions of Theorem \ref{th:mother} are clearly
satisfied since smooth elements with compact support in $\Omega_{0}$
and $\Omega_{1}$, respectively, are already a dense sub-domain as
in the separate cases of Subsection \ref{subsec:Two-Descendants-of}.
From the inclusion \eqref{eq:incl},\eqref{eq:incl-1}, we read off
that the resulting evo-system 
\begin{equation}
\left(\partial_{0}M\left(0\right)+M^{\prime}\left(0\right)+A\right)\left(\begin{array}{c}
v_{\Omega_{0}}+E_{\Omega_{1}}\\
\left(\begin{array}{c}
T_{\Omega_{0}}\\
H_{\Omega_{1}}
\end{array}\right)
\end{array}\right)=\left(\begin{array}{c}
f_{\Omega_{0}}-j_{\mathrm{ext},\Omega_{1}}\\
\left(\begin{array}{c}
g_{\mathrm{sym},\Omega_{0}}\\
k_{\mathrm{ext},\Omega_{1}}
\end{array}\right)
\end{array}\right)\label{eq:evo-mix}
\end{equation}
indeed yields
\[
\partial_{0}\left(\rho_{*,\Omega_{0}}+\epsilon_{\Omega_{1}}\right)\left(v_{\Omega_{0}}+E_{\Omega_{1}}\right)-\Div_{\Omega_{0}}T_{\Omega_{0}}-\curl_{\Omega_{1}}H_{\Omega_{1}}=f_{\Omega_{0}}-j_{\mathrm{ext},\Omega_{1}},
\]
which in turn splits into
\begin{eqnarray*}
\partial_{0}\rho_{*,\Omega_{0}}v_{\Omega_{0}}-\Div_{\Omega_{0}}T_{\Omega_{0}} & = & f_{\Omega_{0}},\\
\partial_{0}\epsilon_{\Omega_{1}}E_{\Omega_{1}}-\curl_{\Omega_{1}}H_{\Omega_{1}} & = & -j_{\mathrm{ext},\Omega_{1}}.
\end{eqnarray*}
The second block row yields another pair of equations
\begin{eqnarray*}
\partial_{0}C^{-1}T_{\Omega_{0}}-\Grad v_{\Omega_{0}} & = & g_{\mathrm{sym},\Omega_{0}},\\
\partial_{0}\mu_{\Omega_{1}}H_{\Omega_{1}}+\curl E_{\Omega_{1}} & = & k_{\mathrm{ext},\Omega_{1}}.
\end{eqnarray*}
The actual system models now natural transmission conditions on the
common boundary part $\dot{\Omega}_{0}\cap\dot{\Omega}_{1}$ and the
homogeneous Dirichlet boundary condition on $\dot{\Omega}_{0}\setminus\dot{\Omega}_{1}$
and the standard homogeneous electric boundary condition on $\dot{\Omega}_{1}\setminus\dot{\Omega}_{0}$
without assuming any smoothness of the boundary.

On the contrary, assuming sufficient regularity of the boundary one
can see that the model yields a generalization of the classical transmission
conditions on $\dot{\Omega}_{0}\cap\dot{\Omega}_{1}$: 
\begin{equation}
\begin{array}{rl}
T_{\Omega_{0}}n & =n\times H_{\Omega_{1}},\\
n\times v_{\Omega_{0}} & =n\times E_{\Omega_{1}},
\end{array}\label{eq:transmission}
\end{equation}
where $n$ is a smooth unit normal field on $\dot{\Omega}_{0}\cap\dot{\Omega}_{1}$.
Indeed, with 
\[
\left(\begin{array}{c}
v_{\Omega_{0}}+E_{\Omega_{1}}\\
\left(\begin{array}{c}
T_{\Omega_{0}}\\
H_{\Omega_{1}}
\end{array}\right)
\end{array}\right)\in D\left(A\right)
\]
we have (noting for the smooth exterior unit normal vector fields
$n_{\dot{\Omega}_{0}}$, $n_{\dot{\Omega}_{1}}$ on the boundaries
of $\Omega_{0}$ and $\Omega_{1}$, respectively, that $n_{\dot{\Omega}_{0}}=-n_{\dot{\Omega}_{1}}$
on $\dot{\Omega}_{0}\cap\dot{\Omega}_{1}$) with 
\[
\tilde{A}=\left(\begin{array}{cc}
0 & \left(-\begin{array}{cc}
\Div_{\Omega_{0}} & -\curl_{\Omega_{1}}\end{array}\right)\\
\left(\begin{array}{c}
-\Grad_{\Omega_{0}}\\
\curl_{\Omega_{1}}
\end{array}\right) & \left(\begin{array}{cc}
0 & 0\\
0 & 0
\end{array}\right)
\end{array}\right),
\]
that
\begin{eqnarray*}
0 & = & \left\langle \left(\begin{array}{c}
v_{\Omega_{0}}+E_{\Omega_{1}}\\
\left(\begin{array}{c}
T_{\Omega_{0}}\\
H_{\Omega_{1}}
\end{array}\right)
\end{array}\right)\Big|A\left(\begin{array}{c}
v_{\Omega_{0}}+E_{\Omega_{1}}\\
\left(\begin{array}{c}
T_{\Omega_{0}}\\
H_{\Omega_{1}}
\end{array}\right)
\end{array}\right)\right\rangle _{H}\\
 & = & \left\langle \left(\begin{array}{c}
v_{\Omega_{0}}+E_{\Omega_{1}}\\
\left(\begin{array}{c}
T_{\Omega_{0}}\\
H_{\Omega_{1}}
\end{array}\right)
\end{array}\right)\Big|\tilde{A}\left(\begin{array}{c}
v_{\Omega_{0}}+E_{\Omega_{1}}\\
\left(\begin{array}{c}
T_{\Omega_{0}}\\
H_{\Omega_{1}}
\end{array}\right)
\end{array}\right)\right\rangle _{H}\\
 & = & -\left\langle v_{\Omega_{0}}\Big|\:\Div\:T_{\Omega_{0}}\right\rangle _{L^{2}\left(\Omega_{0},\mathbb{R}^{3}\right)}-\left\langle T_{\Omega_{0}}\Big|\Grad_{\Omega_{0}}v_{\Omega_{0}}\right\rangle _{L^{2}\left(\Omega_{0},\mathbb{R}^{3\times3}\right)}+\\
 &  & +\left\langle H_{\Omega_{1}}\Big|\curl_{\Omega_{1}}E_{\Omega_{1}}\right\rangle _{L^{2}\left(\Omega_{1},\mathbb{R}^{3}\right)}-\left\langle E_{\Omega_{1}}\Big|\curl_{\Omega_{1}}H_{\Omega_{1}}\right\rangle _{L^{2}\left(\Omega_{1},\mathbb{R}^{3}\right)}\\
 & = & -\int_{\dot{\Omega}_{0}\cap\dot{\Omega}_{1}}v_{\Omega_{0}}^{\top}T_{\Omega_{0}}n_{\dot{\Omega}_{0}}do+\int_{\dot{\Omega}_{0}\cap\dot{\Omega}_{1}}n_{\dot{\Omega}_{1}}^{\top}\left(E_{\Omega_{1}}\times H_{\Omega_{1}}\right)do\\
 & = & -\int_{\dot{\Omega}_{0}\cap\dot{\Omega}_{1}}v_{\Omega_{0}}^{\top}T_{\Omega_{0}}n_{\dot{\Omega}_{0}}do+\int_{\dot{\Omega}_{0}\cap\dot{\Omega}_{1}}E_{\Omega_{1}}^{\top}\left(n_{\dot{\Omega}_{0}}\times H_{\Omega_{1}}\right)do.
\end{eqnarray*}
Since $\left(v_{\Omega_{0}}+E_{\Omega_{1}}\right)\in D\left(\interior\grad\right)$
is by construction admissible we may choose $v_{\Omega_{0}}=E_{\Omega_{1}}$
on the interface and conclude that
\begin{equation}
T_{\Omega_{0}}n_{\dot{\Omega}_{0}}=n_{\dot{\Omega}_{0}}\times H_{\Omega_{1}}\label{eq:trans1}
\end{equation}
is a needed transmission condition. In particular, we see
\[
n_{\dot{\Omega}_{0}}^{\top}T_{\Omega_{0}}n_{\dot{\Omega}_{0}}=0.
\]
Inserting the explicit transmission condition \eqref{eq:trans1} now
yields
\begin{eqnarray*}
-\int_{\dot{\Omega}_{0}\cap\dot{\Omega}_{1}}\left(n_{\dot{\Omega}_{0}}\times\left(n_{\dot{\Omega}_{0}}\times\left(v_{\Omega_{0}}-E_{\Omega_{1}}\right)\right)\right)^{\top}\left(n_{\dot{\Omega}_{0}}\times H_{\Omega_{1}}\right)do=\\
=\int_{\dot{\Omega}_{0}\cap\dot{\Omega}_{1}}\left(v_{\Omega_{0}}-E_{\Omega_{1}}\right)^{\top}\left(n_{\dot{\Omega}_{0}}\times H_{\Omega_{1}}\right)do=0,
\end{eqnarray*}
which, with $n_{\dot{\Omega}_{0}}\times H_{\Omega_{1}}$ for $H_{\Omega_{1}}\in D\left(\curl_{\Omega_{1}}\right)$
being sufficiently arbitrary, now implies
\[
n_{\dot{\Omega}_{0}}\times v_{\Omega_{0}}=n_{\dot{\Omega}_{0}}\times E_{\Omega_{1}}
\]
i.e. the continuity of the tangential components 
\[
v_{\Omega_{0},\mathrm{t}}=E_{\Omega_{1},\mathrm{t}},
\]
as a complementing transmission condition. These more or less heuristic
considerations motivate to take the above evo-system as a appropriate
generalization to cases, where the boundary does \emph{not} have a
reasonable normal vector field.

All in all, we summarize our findings in the following well-posedness
result.

\begin{thm}The evo-system \eqref{eq:evo-mix} is well-posed if $\rho_{*,\Omega_{0}},C_{\Omega_{0}}$
and $\epsilon_{\Omega_{1}},\mu_{\Omega_{1}}$ are selfadjoint, non-negative,
continuous operators on $L^{2}\left(\Omega_{0},\mathbb{R}^{3}\right)$,
$L^{2}\left(\Omega_{0},\mathrm{sym}\left[\mathbb{R}^{3\times3}\right]\right)$
and on $L^{2}\left(\Omega_{1},\mathbb{R}^{3}\right)$, respectively,
$\sigma_{\Omega_{1}}$ is continuous and linear on $L^{2}\left(\Omega_{1},\mathbb{R}^{3}\right)$
and such that
\[
\rho_{*,\Omega_{0}},C_{\Omega_{0}},\mu_{\Omega_{1}}\geq\eta_{0}>0,
\]
as well as
\[
\nu\epsilon_{\Omega_{1}}+\mathrm{sym}\left(\sigma_{\Omega_{1}}\right)\geq\eta_{0}>0
\]
for some real number $\eta_{0}$ and all sufficiently large $\nu$.
\end{thm}

\begin{rem}~

\begin{enumerate}

\item If we formally transcribe the time-harmonic case into its time
dependent form, the transmission conditions of \cite{zbMATH02114868}
are actually 
\begin{equation}
\begin{array}{rl}
T_{\Omega_{0}}n & =n\times\partial_{0}^{-1}H_{\Omega_{1}},\\
n\times\partial_{0}^{-1}v_{\Omega_{0}} & =n\times E_{\Omega_{1}}.
\end{array}\label{eq:transmission-1}
\end{equation}
Although these obviously differ from \eqref{eq:transmission}, we
give preference to our choice above for several reasons. For one,
the energy balance requirement of \cite[formula (5)]{zbMATH02114868},
which reads as
\begin{equation}
v_{\Omega_{0}}^{\top}T_{\Omega_{0}}n=n^{\top}\left(H_{\Omega_{1}}\times E_{\Omega_{1}}\right),\label{eq:energy-balance}
\end{equation}
is satisfied by \eqref{eq:transmission} but not by \eqref{eq:transmission-1}.
With the latter transmission conditions we obtain instead 
\[
v_{\Omega_{0}}^{\top}T_{\Omega_{0}}n=\left(\partial_{0}E_{\Omega_{1}}\right)^{\top}\left(n\times\left(\partial_{0}^{-1}H_{\Omega_{1}}\right)\right)=n^{\top}\left(\left(\partial_{0}^{-1}H_{\Omega_{1}}\right)\times\left(\partial_{0}E_{\Omega_{1}}\right)\right).
\]
The problem seems to be that the difference to \eqref{eq:energy-balance}
becomes unnoticeable in the formal time-harmonic transcription of
\cite{zbMATH02114868}, since there $\partial_{0}$ is formally replaced
by $\ii\omega\sqrt{\epsilon_{0}\mu_{0}}$ and so algebraic cancellation
essentially makes the product rule for differentiation disappear,
erroneously suggesting that the energy balance\footnote{The correct energy balance in the time-harmonic case would actually
involve temporal convolution products.} is satisfied.

\item In the notation above, \eqref{eq:incl}, \eqref{eq:M0-mx},
\eqref{eq:M1-mx}, if $M\left(0\right)$ is already strictly positive
definite, we can construct a fundamental solution as a small perturbation
of the fundamental solution of $\partial_{0}+\sqrt{M\left(0\right)}^{-1}A\sqrt{M\left(0\right)}^{-1}$,
which in turn is obtained from the unitary group 
\[
\left(\exp\left(-t\,\sqrt{M\left(0\right)}^{-1}A\sqrt{M\left(0\right)}^{-1}\right)\right)_{t\in\mathbb{R}}
\]
 by cut-off as 
\[
\left(\chi_{_{\lci0\infty}}\left(t\right)\;\exp\left(-t\,\sqrt{M\left(0\right)}^{-1}A\sqrt{M\left(0\right)}^{-1}\right)\right)_{t\in\mathbb{R}}.
\]
The restriction of the fundamental solution to $\lci0\infty$ yields
the family 
\[
\left(\exp\left(-t\,\sqrt{M\left(0\right)}^{-1}A\sqrt{M\left(0\right)}^{-1}\right)\right)_{t\in\lci0\infty}
\]
commonly referred to as the associated one-parameter semi-group. In
general, however, a fundamental solution may be complicated or impossible
to construct.

\item We note that beyond eddy current type behavior, which is actually
a change of type situation from hyperbolic to parabolic, and beyond
the possibility of including for example piezo-electric effects via
a more complex material law, we may actually allow for completely
general rational material laws as long as condition \eqref{eq:posdef}
is warranted.

\end{enumerate}\end{rem}

\bibliographystyle{plain}
\end{document}